\documentclass[11pt,onecolumn,twoside]{IEEEtran}
\usepackage{amsmath ,amssymb ,euscript ,yfonts
,psfrag,latexsym,dsfont,graphicx,bbm,color,amstext,wasysym}
\usepackage[usenames,dvipsnames]{xcolor}
\newtheorem{thm}{Theorem}

\newtheorem{prop}[thm]{Proposition}

\newtheorem{defn}[thm]{Definition}

\newcommand{\mR}{{\mathbb R}}
\newcommand{\mC}{{\mathbb C}}

\newcommand{\trace}{\operatorname{tr}}
\newcommand{\tr}{\operatorname{tr}}

\newcommand{\bm}{{\boldsymbol m}}

\newcommand{\bM}{{\boldsymbol M}}
\newcommand{\brho}{{\boldsymbol\rho}}
\newcommand{\bmu}{{\boldsymbol\mu}}

\newcommand{\bT}{{\boldsymbol{\mathcal T}}}
\newcommand{\cT}{{\mathcal T}}

\newcommand{\cH}{{\mathcal H}}
\newcommand{\cS}{{\mathcal S}}
\newcommand{\cF}{{\mathcal F}}
\newcommand{\cL}{{\mathcal L}}
\newcommand{\trA}{{\trace_{0}}}
\newcommand{\trB}{{\trace_{1}}}

\newcommand{\jj}{{\rm j}}
\newcommand{\bartrA}{{\underline{\trace}_{0}}}
\newcommand{\bartrB}{{\underline{\trace}_{1}}}

\begin{document}
\title{\huge \bf  Matrix-valued Monge-Kantorovich\\ Optimal Mass Transport}

\author{Lipeng Ning, Tryphon T. Georgiou and
Allen Tannenbaum
\thanks{\hspace*{-10pt} L. Ning and T.T. Georgiou are with the Department of Electrical \& Computer Engineering,
University of Minnesota, Minneapolis, MN 55455, {\tt $\{$ningx015,tryphon$\}$@umn.edu}, A. Tannenbaum is with the Comprehensive
Cancer Center and the Department of Electrical \& Computer Engineering, University of
Alabama, Birmingham, AL 35294, {\tt tannenba@uab.edu}}
}

\maketitle

\begin{abstract} We formulate an optimal transport problem for matrix-valued density functions. This is pertinent in the
spectral analysis of multivariable time-series. The ``mass'' represents energy at various frequencies whereas, in addition to a usual
transportation cost across frequencies, a cost of rotation is also taken into account.
We show that it is natural to seek the transportation plan in the tensor product of the spaces for the two matrix-valued marginals.
In contrast to the classical Monge-Kantorovich setting, the transportation plan is no longer supported on a thin zero-measure set.
\end{abstract}

\section{Introduction}

The formulation of optimal mass transport (OMT) goes back to the
work of G. Monge in 1781 \cite{Monge1781}. The modern formulation is
due to Kantorovich in 1947 \cite{Kantorovich1942}. In recent years
the subject is evolving rather rapidly due to the wide range of
applications in economics, theoretical physics, probability, etc.
Important recent monographs on the subject include
\cite{Villani_book,Ambrosio,Rachev1998mass}.

Our interest in the subject of matrix-valued transport originates in
the spectral analysis of multi-variable time-series. It is natural
to consider the weak topology for power spectra. This is because
statistics typically represent integrals of power spectra and hence
a suitable form of continuity is desirable. Optimal mass transport
and the geometry of the Wasserstein metric provide a natural
framework for studying scalar densities. Thus, the scalar OMT theory was
used in \cite{Jiang2012} for modeling slowly time-varying changes in the power spectra of time-series.
The salient feature of matrix-valued densities is that power can shift across frequencies as well as across different channels via rotation of
the corresponding eigenvectors. Thus, transport between
matrix-valued densities requires that we take into account the cost
of rotation as well as the cost of shifting power across frequencies.

Besides the formulation of a ``non-commutative'' Monge-Kantorovich
transportation problem, the main results in the paper are that (1)
the solution to our problem can be cast as a convex-optimization problem, (2) geodesics can be determined by convex programming, and
(3) that the optimal transport plan has support which, in contrast to
the classical Monge-Kantorovich setting, is no longer contained
on a thin zero-measure set.

\section{Preliminaries on Optimal Mass Transport}

Consider two probability density functions $\mu_0$ and $\mu_1$ supported on $\mR$.
Let $M(\mu_0,\mu_1)$ be the set of probability measures $m(x,y)$ on $\mR\times \mR$ with $\mu_0$ and $\mu_1$ as marginal density functions, i.e.
\[
\int_{\mR} m(x,y)dy=\mu_0(x),~ \int_{\mR} m(x,y) dx=\mu_1(y),~m(x,y)\geq0.
\]
The set $M(\mu_0,\mu_1)$ is not empty since
$m(x,y)=\mu_0(x)\mu_1(y)$ is always a feasible solution. Probability
densities can be thought of as distributions of mass and a cost
$c(x,y)$ associated with transferring one unit of mass from one
location $x$ to $y$. For $c(x,y)=|x-y|^2$ the optimal transport cost
gives rise to the 2-Wasserstein metric
\[
W_2(\mu_0, \mu_1)=\cT_2(\mu_0, \mu_1)^\frac12
\]
where
\begin{align}\label{scalar:primal}
\cT_2(\mu_0, \mu_1):=\inf_{m\in M(\mu_0,\mu_1)} \int_{\mR\times \mR} c(x,y)m(x,y) dxdy.
\end{align}
Problem \eqref{scalar:primal} is a linear programming problem with dual
\begin{align}\label{scalar:dual1}
\sup_{\phi,\psi} \bigg\{\int_{\mR}  \phi_0 \mu_0-\phi_1 \mu_1 dx \mid&~  \phi_0(x)-\phi_1(y)\leq c(x,y) \bigg\}
\end{align}
see e.g., \cite{Villani_book}. Moreover,  for the quadratic cost
function $c(x,y)=|x-y|^2$, $\cT_2(\mu_0,\mu_1)$ can also be written
explicitly in term of the cumulative distributions functions
\begin{align*}
F_i(x)=&~\int_{-\infty}^x \mu_idx \text{~for~} i=0,1,
\end{align*}
as follows (see \cite[page 75]{Villani_book})
  \begin{align}\label{eq:scalarcumulative}
  \cT_2(\mu_0,\mu_1)=\int_0^1 |F_0^{-1}(t)-F_1^{-1}(t)|^2 dt,
  \end{align}
and  the optimal joint probability density $m\in M(\mu_0, \mu_1)$
has support on $(x,T(x))$ where $T(x)$ is the sub-differential of a convex lower semi-continuous function. More specifically,
$T(x)$ is uniquely defined by
\begin{align}\label{eq:cumulative}
F_0(x)=F_1(T(x)).
\end{align}
%
Finally, a geodesic $\mu_\tau (\tau\in[0, 1])$ between $\mu_0$ and $\mu_1$ can be written explicitly in terms of the cumulative function
$F_{\tau}$ defined by
\begin{align}\label{eq:scalarGeodesic}
F_\tau((1-\tau)x+\tau T(x))=F_0(x).
\end{align}
Then, clearly,
\begin{align*}
W_2(\mu_0, \mu_\tau)=&~\tau W_2 (\mu_0, \mu_1)\\
W_2(\mu_\tau, \mu_1)=&~(1-\tau)W_2 (\mu_0, \mu_1).
\end{align*}

%

\section{Matrix-valued  Optimal Mass Transport}
We consider the family
\begin{align*}
\cF:=\bigg\{\bmu \mid \mbox{for }x\in\mR, \bmu(x)\in\mC^{n\times n}\mbox{ Hermitian}, \bmu(x)\geq0,\; \trace(\int_{\mR} \bmu(x) dx)=1 \bigg\},
\end{align*}
of Hermitian positive semi-definite, matrix-valued densities on
$\mR$, normalized so that their trace integrates to $1$. We motivate
a transportation cost to this matrix-valued setting and introduce a
generalization of the Monge-Kantorovich OMT to matrix-valued
densities.

\subsection{Tensor product and partial trace}

Consider two $n$-dimensional Hilbert spaces $\cH_0$ and $\cH_1$ with
basis $\{u_1, \ldots, u_n\}$ and $\{v_1, \dots, v_n\}$,
respectively. Let $\cL(\cH_0)$ and $\cL(\cH_1)$ denote the space of
linear operators on $\cH_0$ and $\cH_1$, respectively. For
$\brho_0\in\cL(\cH_0)$ and $\brho_1\in\cL(\cH_1)$, we denote their
tensor product by $\brho_0 \otimes \brho_1 \in\cL(\cH_0\otimes
\cH_1).$ Formally, the latter is defined via
\[
\brho_0 \otimes \brho_1 \;:\; u\otimes v\mapsto  \brho_0 u\otimes \brho_1v.
\]
Since our spaces are finite-dimensional this is precisely the Kronecker product of the corresponding matrix representation of the two operators.

Consider $\brho\in \cL(\cH_0\otimes \cH_1)$ which can be thought of
as a matrix of size $n^2\times n^2$. The partial traces
$\tr_{\cH_0}$ and $\tr_{\cH_1}$, or $\tr_0$ and $\tr_1$ for brevity,
are linear maps
\begin{eqnarray*}
 \brho \in \cL(\cH_0\otimes \cH_1) &\mapsto& \tr_1(\brho)\in \cL(\cH_0)\\
 &\mapsto& \tr_0(\brho)\in \cL(\cH_1)
\end{eqnarray*}
that are defined as follows.
Partition $\brho$ into $n\times n$ block-entries and denote by $\brho_{k\ell}$ the $(k,\ell)$-th block ($1\leq k,\ell\leq n$).
Then the partial trace, e.g.,
\[
\brho_0:=\tr_1(\brho)
\]
is the $n\times n$ matrix with
\[
[\brho_0]_{k\ell}=\trace(\brho_{k\ell}),~\text{for~} 1\leq k,\ell\leq n.
\]
The partial trace
\[
\brho_1:=\tr_0(\brho)
\]
is defined in a similar manner for a corresponding partition of
$\brho$, see e.g., \cite{Petz2008quantum}. More specifically, for
$1\leq i,j \leq n$, let $\brho^{ij}$ be a sub-matrix of $\brho$ of
size $n\times n$ with the $(k, \ell)$-th entry $
[\brho^{ij}]_{k\ell}=[\brho_{k\ell}]_{ij} $. Then the $(i,j)$-th
entry of $\brho_1$ is
\[
[\brho_1]_{ij}=\tr(\brho^{ij}).
\]
Thus
\[
\tr_1(\brho_0\otimes \brho_1)=\tr(\brho_1) \brho_0  \text{~and~} \tr_0(\brho_0\otimes \brho_1)=\tr(\brho_0) \brho_1.
\]


\subsection{Joint density for matrix-valued distributions}
A naive attempt to define a joint probability density given
marginals $\bmu_0, \bmu_1\in\cF_n$ is to consider a matrix-valued
density with support on $\mR\times\mR$ such that $\bm\geq 0$ and
\begin{align}\label{eq:jointNaive}
~\int_{\mR} \bm(x,y)dy= \bmu_0(x),
~\int_{\mR} \bm(x,y)dx= \bmu_1(y).
\end{align}
However, in contrast to the scalar case, this constraint is not always feasible. To see this
consider
\begin{align*}
&~\bmu_0(x)=\left[\begin{array}{cc}\frac12 & 0\\ 0 &0  \end{array}    \right]\delta(x-x_1)+\left[\begin{array}{cc}0 & 0\\ 0 &\frac12  \end{array}    \right]\delta(x-x_2),\\
&~\bmu_1(x)=\left[\begin{array}{cc}~~\frac14 & -\frac14\\ -\frac14 & ~~\frac14 \end{array}    \right]\delta(x-x_1)+\left[\begin{array}{cc}\frac14 & \frac14\\ \frac14 &\frac14  \end{array}    \right]\delta(x-x_2).
\end{align*}
It is easy to show that \eqref{eq:jointNaive} cannot be met.

A natural definition for joint densities $\bm$ that can serve as a
transportation plan may be defined as follows. For
$(x,y)\in\mR\times\mR$,
\begin{subequations}\label{TransportPlan}
\begin{align}
&~\bm(x,y) \text{~is $n^2\times n^2$ positive semi-definite matrix,}\label{eq:TransportPlanA}
\end{align}
and with
\begin{align}
&~\bm_0(x,y):=\trB(\bm(x,y)),
\bm_1(x,y):=\trA(\bm(x,y)),\label{eq:TransportPlanB}
\end{align}
one has
\begin{align}
&~\int_{\mR} \bm_0(x,y) dy= \bmu_0(x), \int_{\mR} \bm_1(x,y)dx= \bmu_1(y).\label{eq:TransportPlanC}
\end{align}
\end{subequations}
Thus, we denote by
\begin{align*}
\bM(\bmu_0,\bmu_1):= \Big\{\bm \mid&~ \eqref{eq:TransportPlanA}-\eqref{eq:TransportPlanC} \text{~are satisfied}\Big\}.
\end{align*}
For this family, given marginals, there is always an admissible
joint distribution as stated in the following proposition.

\begin{prop}
For any $\bmu_0, \bmu_1\in\cF_n$, the set $\bM(\bmu_0,\bmu_1)$ is not empty.
\end{prop}
\begin{proof}
Clearly, $
\bm:=\bmu_0\otimes\bmu_1
\in \bM(\bmu_0,\bmu_1).$
\end{proof}

We next motivate a natural form for the transportation cost. This is a functional on the joint density as in the scalar case.
However, besides a penalty on ``linear'' transport we now take into account an ``angular'' penalty as well.

\subsection{Transportation cost}

We interpret $\trace(\bm(x,y))$ as the amount of ``mass'' that is
being transferred from $x$ to $y$. Thus, for a scalar cost function
$c(x,y)$ as before, one may simply consider
\begin{align}\label{eq:matrixprimala}
\min_{\bm\in \bM(\bmu_0,\bmu_1)} \int_{\mR\times \mR} c(x,y)\tr(\bm(x,y)) dxdy.
\end{align}
However, if $\tr(\bmu_0(x))=\tr(\bmu_1(x))~\forall x\in\mR$, then the optimal value of \eqref{eq:matrixprimala} is zero.
Thus \eqref{eq:matrixprimala} fails to quantify mismatch in the matricial setting.

For simplicity, throughout, we only consider marginals $\bmu$, which
pointwise satisfy $\tr(\bmu)>0$.
$\tr(\bmu(x))$ is a scalar-valued density representing mass at location $x$
while $\frac{\bmu(x)}{\tr(\bmu(x))}$ has trace $1$ and contains directional information.
Likewise, for a
joint density $\bm(x,y)$,
assuming $\bm(x,y)\neq 0$, we consider
\begin{align*}
\bartrA(\bm(x,y)):=\trA(\bm(x,y))/\trace(\bm(x,y)) \\
\bartrB(\bm(x,y)):=\trB(\bm(x,y))/\trace(\bm(x,y)).
\end{align*}
Since $\bartrA(\bm(x,y))$ and $\bartrB(\bm(x,y))$ are normalized to
have unit trace, their difference captures the directional mismatch
between the two partial traces. Thus take
\[
\trace(\| (\bartrA- \bartrB)\bm(x,y)\|_{\rm F}^2 \bm(x,y))
\]
to quantify the rotational mismatch. The above motivates the
following cost functional that includes both terms, rotational and
linear:
\[
\trace\bigg( (c(x,y)+\lambda\| (\bartrA-\bartrB) \bm(x,y)\|_{\rm F}^2) \bm(x,y)\bigg)
\]
where $\lambda>0$ can be used to weigh in the relative significance of the two terms.

\subsection{Optimal transportation problem}\label{sec:convex}

In view of the above, we now arrive at the following formulation of
a matrix-valued version of the OMT, namely the determination of
\begin{align}\label{prob:MatrixOMT}
\bT_{2,\lambda}(\bmu_0,\bmu_1) :=\min_{\bm\in\bM(\bmu_0,\bmu_1)} \int_{\mR\times \mR} \trace\bigg( (c+\lambda\| (\bartrA-\bartrB)\bm \|_{\rm F}^2) \bm\bigg) dxdy.
\end{align}
Interestingly, \eqref{prob:MatrixOMT} can be cast as a convex optimization problem. We explain this next.

Since, by definition,
\begin{eqnarray*}
\bartrA(\bm)\trace(\bm)&=&\trA(\bm),\\
\bartrB(\bm)\trace(\bm)&=&\trB(\bm),
\end{eqnarray*}
we deduce that
\begin{align*}
\|(\bartrA-\bartrB)\bm\|_{\rm F}^2\trace(\bm)=&~\frac{\|(\bartrA-\bartrB)\bm\|_{\rm F}^2\trace(\bm)^2}{\trace(\bm)}\\
=&~\frac{\|(\tr_0-\tr_1)\bm\|_{\rm F}^2}{\trace(\bm)}.
\end{align*}
Now let $m(x,y)=\trace(\bm(x,y))$ and let $\bm_0(x,y)$ and
$\bm_1(x,y)$ be as in \eqref{TransportPlan}. The expression for the
optimal cost in \eqref{prob:MatrixOMT} can be equivalently written
as
\begin{align}\label{eq:second}
\min_{\bm_0,\bm_1, m} \Big\{ \int~ \left(c(x,y)m(x,y) +\lambda \frac{\|\bm_0-\bm_1\|_{\rm F}^2}{m}\right) dxdy \mid
&~\bm_0(x,y), ~\bm_1(x,y)\geq 0\nonumber\\
&~\trace(\bm_0(x,y))=\trace(\bm_1(x,y))=m(x,y)\nonumber\\
&\int \bm_0(x,y)dy=\bmu_0(x)\nonumber\\
&\int \bm_1(x,y)dx=\bmu_1(y)\Big\}.
\end{align}
Since, for $x>0$,
\[
\frac{(y-z)^2}{x}
\] is convex in the arguments $x,y,z$, it readily follows that the integral in \eqref{eq:second} is a convex functional. All
constraints in \eqref{eq:second} are also convex and therefore, so is the optimization problem.

\section{On the geometry of  Optimal Mass Transport}

A standard result in the (scalar) OMT theory is that the
transportation plan is the sub-differential of a convex function. As
a consequence the transportation plan has support only on a
monotonically non-decreasing zero-measure set. This is no longer
true for the optimal transportation plan for matrix-valued density
functions and this we discuss next.

In optimal transport theory for scalar-valued distributions, the
optimal transportation plan has a certain cyclically monotonic
property \cite{Villani_book}. More specifically, if $(x_1, y_1)$,
$(x_2, y_2)$ are two points where the transportation plan has
support, then $x_2>x_1$ implies $y_2\geq y_1$. The interpretation is
that optimal transportation paths do not cross. For the case of
matrix-valued distributions as in \eqref{prop:metric}, this property
may not hold in the same way. However, interestingly, a weaker
monotonicity property holds for the supporting set of the optimal
matrix transportation plan. The property is defined next and the precise
statement is given in Proposition \ref{thm:M} below.

\begin{defn}
A set $\cS\subset \mR^2$ is called a {\em $\lambda$-monotonically
non-decreasing}, for $\lambda>0$, if for any two points $(x_1, y_1),
(x_2, y_2)\in \cS$, it holds that
\[
(x_2-x_1)(y_1-y_2)\leq \lambda.
\]
\end{defn}
\vspace*{.3cm} A geometric interpretation for a
$\lambda$-monotonically non-decreasing set is that if $(x_1, y_1)$,
$(x_2, y_2)\in \cS$ and $x_2>x_1$, $y_1>y_2$, then the area of the
rectangle with vertices $(x_i,y_j)$ ($i,j\in\{1,2\}$) is not larger
than $\lambda$. The transportation plan of the scalar-valued optimal
transportation problem with a quadratic cost has support on a
$0$-monotonically non-decreasing set.

\begin{prop} \label{thm:M}
Given $\bmu_0, \bmu_1\in \cF$, let $\bm$ be the optimal
transportation plan in \eqref{prob:MatrixOMT} with $\lambda>0$. Then
$\bm$ has support on at most a $(4\cdot \lambda)$-monotonically non-decreasing
set.
\end{prop}
\begin{proof}
See the appendix.
\end{proof}

Then the optimal transportation cost $\cT_{2,\lambda}(\bmu_0, \bmu_1)$ satisfies the following
properties:
\begin{enumerate}
\item $\cT_{2,\lambda}(\bmu_0, \bmu_1)$=$\cT_{2,\lambda}(\bmu_1,\bmu_0)$,
\item $\cT_{2,\lambda}(\bmu_0, \bmu_1)\geq0$,
\item $\cT_{2,\lambda}(\bmu_0,\bmu_1)=0$ if and only if $\bmu_0=\bmu_1$.
\end{enumerate}
Thus, although $\cT_{2,\lambda}(\bmu_0, \bmu_1)$  can be used to compare matrix-valued
densities, it is not a metric and neither is
$\cT_{2,\lambda}^\frac12$ since the triangular inequality does not hold in general.
We will introduce a slightly different
formulation of a transportation problem which does give rise to a
metric.

\subsection{Optimal transport on a subset}
In this subsection, we restrict attention to a certain subset of transport
plans $\bM(\bmu_0, \bmu_1)$ and show that the corresponding optimal
transportation cost induces a metric. More specifically, let
\begin{align*}
\bM_0(\bmu_0,\bmu_1):=\bigg\{\bm \mid~ \bm(x,y)=\bmu_0(x)\otimes\bmu_1(y) a(x,y),~ \bm \in\bM\bigg\}.
\end{align*}
For $\bm(x,y)\in \bM_0(\bmu_0, \bmu_1)$,
\begin{align*}
\bartrA(\bm(x,y)):=\bmu_1(x)/\trace(\bmu_1(x)) \\
\bartrB(\bm(x,y)):=\bmu_0(y)/\trace(\bmu_0(y)).
\end{align*}
Given $\bmu_0$ and $\bmu_1$, the ``orientation'' of the mass of $\bm(x,y)$ is fixed.
Thus, in this case, the optimal transportation cost is
\begin{align}\label{cost:TlambdaB}
\tilde\bT_{2,\lambda}(\bmu_0, \bmu_1) :=
\min_{\bm \in \bM_0(\bmu_0,\bmu_1)} \int \trace\bigg( (c+\lambda\| (\bartrA-\bartrB) \bm(x,y)\|_{\rm F}^2) \bm\bigg) dxdy.
\end{align}

\begin{prop}\label{prop:metric}
For $\bT_{2,\lambda}$ as in \eqref{cost:TlambdaB} and $\bmu_0, \bmu_1\in \cF$,
\begin{equation}\label{eq:metric}
d_{2,\lambda}(\bmu_0, \bmu_1):=\left(\tilde \bT_{2,\lambda}(\bmu_0, \bmu_1)\right)^{\frac12}
\end{equation}
defines a metric on $\cF$.
\end{prop}

\begin{proof}
It is straightforward to prove that
 \[
 d_{2,\lambda}(\bmu_0,\bmu_1)=d_{2,\lambda}(\bmu_1,\bmu_0)\geq0
 \]
 and that
 $d_{2,\lambda}(\bmu_0,\bmu_1)=0$ if and only if $\bmu_0=\bmu_1$. We will show that the triangle inequality also holds.
 For $\bmu_0, \bmu_1, \bmu_2\in \cF_n$, let
 \begin{align*}
\bm_{01}(x,y)=&~\frac{\bmu_{0}(x)}{\trace(\bmu_0(x))}\otimes\frac{\bmu_{1}(y)}{\trace(\bmu_1(y))} m_{01}(x,y)\\
\bm_{12}(y,z)=&~\frac{\bmu_{1}(y)}{\trace(\bmu_1(y))}\otimes\frac{\bmu_{2}(z)}{\trace(\bmu_2(z))} m_{12}(y,z)
 \end{align*}
 be the optimal transportation plan for the pairs $(\bmu_0, \bmu_1)$ and $(\bmu_1, \bmu_2)$, respectively,
 where $m_{01}$ and $m_{12}$ are two (scalar-valued) joint densities on $\mR^2$ with marginals $\trace(\bmu_0)$, $\trace(\bmu_1)$
 and $\trace(\bmu_1)$, $\trace(\bmu_2)$, respectively. Given $m_{01}(x,y)$ and $m_{12}(y,z)$ there is a joint density
 function $m(x, y, z)$ on $\mR^3$ with $m_{01}$ and $m_{12}$ as the marginals on the corresponding subspaces \cite[page 208]{Villani_book}. We denote
 \[
 \bm(x,y,z)=\frac{\bmu_{0}(x)}{\trace(\bmu_0(x))}\otimes\frac{\bmu_{1}(y)}{\trace(\bmu_1(y))}\otimes\frac{\bmu_{2}(z)}{\trace(\bmu_2(z))} m(x,y,z)
 \]
 then it has $\bm_{01}$ and $\bm_{12}$ as the matrix-valued marginal distributions.

 Now, let $\bm_{02}(x,z)=\frac{\bmu_{0}(x)}{\trace\bmu_0(x)}\otimes\frac{\bmu_{2}(z)}{\trace\bmu_2(z)} m_{02}(x,z)$ be the marginal distribution of $\bm(x,y,z)$ when tracing out the $y$-component. Then $\bm_{02}(x,z)$ is a candidate transportation plan between $\bmu_0$ and $\bmu_2$. Thus
 \begin{align*}
 d_{2,\lambda}(\bmu_0,\bmu_2)
  \leq &\left(\int_{\mR^2} \left((x-z)^2+\lambda \|\frac{\bmu_0(x)}{\trace\bmu_0(x)}-\frac{\bmu_2(z)}{\trace\bmu_2(z)} \|_{\rm F}^2\right) m_{02} dxdz\right)^\frac12\\
  =&\left(\int_{\mR^3} \left((x-z)^2+\lambda \|\frac{\bmu_0(x)}{\trace\bmu_0(x)}-\frac{\bmu_2(z)}{\trace\bmu_2(z)} \|_{\rm F}^2\right) m dxdydz\right)^\frac12\\
  =&\bigg(\int_{\mR^3} \Big((x-y+y-z)^2+\lambda \|\frac{\bmu_0(x)}{\trace\bmu_0(x)}-\frac{\bmu_1(y)}{\trace\bmu_1(y)}+
  \frac{\bmu_1(y)}{\trace\bmu_1(y)}-\frac{\bmu_2(z)}{\trace\bmu_2(z)} \|_{\rm F}^2\Big) m dxdydz\bigg)^\frac12\\
 \leq& \bigg(\int_{\mR^2} \Big((x-y)^2+\lambda \|\frac{\bmu_0(x)}{\trace\bmu_0(x)}-\frac{\bmu_1(y)}{\trace\bmu_1(y)} \|_{\rm F}^2\Big)m_{01} dxdy \bigg)^\frac12+\\
 &~~ \bigg(\int_{\mR^2} \Big((y-z)^2+\lambda \|\frac{\bmu_1(y)}{\trace\bmu_1(y)}-\frac{\bmu_2(z)}{\trace\bmu_2(z)} \|_{\rm F}^2\Big)m_{12} dydz \bigg)^\frac12\\
 =&~ d_{2,\lambda}(\bmu_0,\bmu_1)+d_{2,\lambda}(\bmu_1,\bmu_2)
\end{align*}
where the last inequality is from the fact that $L_2$-norm defines a metric.
\end{proof}

\begin{prop} \label{thm:M0}
Given $\bmu_0, \bmu_1\in \cF$, let $\bm$ be the optimal transportation plan in \eqref{eq:metric},
then $\bm$ has support on at most a $(2\cdot \lambda)$-monotonically non-decreasing set.
\end{prop}
\begin{proof}
We need to prove that if $\bm(x_1, y_1)\neq 0$ and $\bm(x_2, y_2)\neq 0$, then $x_2>x_1$, $y_1>y_2$ implies
\begin{align}\label{eq:area2}
(y_1-y_2)(x_2-x_1)\leq2\lambda.
\end{align}
Assume that $\bm$ evaluated at the four points $(x_i,y_j)$, with $i, j\in\{1,2\}$, is as follows
\begin{align*}
\bm(x_i,y_j)=m_{ij}\cdot A_i\otimes B_j
\end{align*}
with
\[
A_i=\frac{\bmu_0(x_i)}{\trace(\bmu_1(x_i))},
~B_i=\frac{\bmu_0(y_i)}{\trace(\bmu_1(y_i))},
\]
and $m_{11}, m_{22}>0$.
The steps of the proof are similar to those of Proposition \ref{thm:M}: first, we
assume that Proposition \ref{thm:M0} fails and that
\[
(y_1-y_2)(x_2-x_1)>2\lambda.
\]
Then we show that a smaller cost can be obtained by rearranging the ``mass''.
Consider the situation when $m_{22}\geq m_{11}$ first and let $\hat \bm$ be a new transportation plan with
\begin{align*}
\hat{\bm}(x_1, y_1)=&~0\\
\hat{\bm}(x_1,y_2)=&~(m_{11}+m_{12})\cdot A_1\otimes B_2 \\
\hat{\bm}(x_2,y_1)=&~(m_{11}+m_{21})\cdot A_2\otimes B_1\\
\hat{\bm}(x_2,y_2)=&~(m_{22}-m_{11})\cdot A_2\otimes B_2 .
\end{align*}
Then, $\hat \bm$ has the same marginals as $\bm$ at the four points and the cost incurred by $\bm$ is
\begin{equation}\label{eq:cost:m}
\sum_{i=1}^2\sum_{j=1}^2 m_{ij}\left((x_i-y_j)^2+\lambda \|A_i-B_j \|_{\rm F}^2 \right)
\end{equation}
while the cost incurred by $\hat\bm$ is
\begin{align}
&(m_{11}+m_{12})\left((x_1-y_2)^2+\lambda \|A_1-B_2\|_{\rm F}^2 \right)\nonumber\\
+&(m_{11}+m_{21})\left((x_2-y_1)^2+\lambda \|A_2-B_1\|_{\rm F}^2 \right)\nonumber\\
+&(m_{22}-m_{11})\left((x_2-y_2)^2+\lambda \|A_2-B_2\|_{\rm F}^2 \right).\label{eq:cost:hatm}
\end{align}
After canceling the common terms, to show that \eqref{eq:cost:m} is larger than \eqref{eq:cost:hatm}, it suffices to show that
\begin{align*}
(y_1-x_1)^2+(y_2-x_2)^2+\lambda\|A_1-B_1 \|_{\rm F}^2+\lambda\|A_2-B_2 \|_{\rm F}^2\nonumber\\
\geq (y_2-x_1)^2+(y_1-x_2)^2+\lambda\|A_1-B_2 \|_{\rm F}^2+\lambda\|A_2-B_1\|_{\rm F}^2.
\end{align*}
The above holds since
\begin{align*}
&(y_1-x_1)^2+(y_2-x_2)^2+\lambda\|A_1-B_1 \|_{\rm F}^2+\lambda\|A_2-B_2 \|_{\rm F}^2\\
\geq&(y_1-x_1)^2+(y_2-x_2)^2\\
=&(y_1-x_2)^2+(y_2-x_1)^2+2(x_2-x_1)(y_1-y_2)\\
>&(y_1-x_2)^2+(y_2-x_1)^2+4\lambda\\
\geq&(y_1-x_2)^2+(y_1-x_2)^2+\lambda(\|A_1-B_2 \|_{\rm F}^2+\|A_2-B_1 \|_{\rm F}^2).
\end{align*}
The case $m_{11}> m_{22}$ proceeds similarly.
\end{proof}

%

\section{Example}

We highlight the relevance of the matrix-valued OMT to spectral analysis by presenting an numerical example of spectral morphing.
The idea is to model slowly time-varying changes in the spectral domain by geodesics in a suitable geometry (see e.g.,
 \cite{Jiang2012,JNG}). The importance of OMT stems from the fact that it induces a weakly continuous metric. Thereby, geodesics smoothly shift spectral power across frequencies lessening the possibility of a fade-in fade-out phenomenon. The classical theory of OMT allows constructing such geodesics for scalar-valued distributions. The example below demonstrates that we can now have analogous construction of geodesics of matrix-valued power spectra as well.

Starting with $\bmu_0, \bmu_1\in\cF$ we approximate the geodesic between them by identifying $N-1$ points between the two. More specifically, we set
$\bmu_{\tau_0}=\bmu_0$ and $\bmu_{\tau_N}=\bmu_1$, and determine
$\bmu_{\tau_k}\in \cF_n$ for $k=1, \ldots, N-1$ by solving
\begin{equation}\label{eq:interp}
\min_{\bmu_{\tau_k},0<k<N} \sum_{k=0}^{N-1} \cT_{2,\lambda}(\bmu_{\tau_{k+1}}, \bmu_{\tau_{k}}).
\end{equation}
As noted in Section~\ref{sec:convex}, numerically this can be solved via a convex
programming problem. The numerical example is based on the following
two matrix-valued power spectral densities
\begin{align*}
\bmu_0&=\left[
      \begin{array}{cc}
        1 & 0 \\
        0.2e^{-\jj\theta} & 1 \\
      \end{array}
    \right]\left[
      \begin{array}{cc}
        \frac{1}{|a_0(e^{\jj\theta})|^2} & 0 \\
        0 & 0.01 \\
      \end{array}
    \right]\left[
      \begin{array}{cc}
        1 & 0.2e^{\jj\theta}\\
        0 & 1 \\
      \end{array}
    \right]\\
\bmu_1&=\left[
      \begin{array}{cc}
        1 & 0.2 \\
        0 & 1 \\
      \end{array}
    \right]\left[
      \begin{array}{cc}
        0.01 & 0 \\
        0 & \frac{1}{|a_1(e^{\jj\theta})|^2} \\
      \end{array}
    \right]\left[
      \begin{array}{cc}
        1 & 0 \\
        0.2 & 1 \\
      \end{array}
    \right]
\end{align*}
with
\begin{align*}
a_0(z)=~&(z^2-1.8\cos(\frac{\pi}{4})z+0.9^2)\\
&(z^2-1.4\cos(\frac{\pi}{3})z+0.7^2)\\
a_1(z)=~&(z^2-1.8\cos(\frac{\pi}{6})z+0.9^2)\\
&(z^2-1.5\cos(\frac{2\pi}{15})z+0.75^2),
\end{align*}
shown in Figure \ref{fig:Mui}. The value of a power spectral density at
each point in frequency is a $2\times 2$ Hermitian matrix. Hence, the $(1, 1)$,
$(1, 2)$, and $(2, 2)$ subplots display the magnitude of the corresponding entries, i.e.,
$|\bmu(1,1)|$, $|\bmu(1,2)|$ ($=|\bmu(2,1)|$) and
$|\bmu(2,2)|$, respectively.  The $(2,1)$
subplot displays the phase $\angle \bmu(1,2)$ ($= -\angle \bmu(2,1)$).

The three dimensional plots in Figure \ref{fig:3D} show the
solution of \eqref{eq:interp} with $\lambda=0.1$ which is an approximation of a geodesic. The
two boundary plots represent the power spectra $\bmu_0$ and $\bmu_1$ shown
in blue and red, respectively, using the same convention about magnitudes and phases.
There are in total $7$ power spectra
$\bmu_{\tau_k}$, $k=1, \ldots, 7$ shown along the geodesic between $\bmu_0$ and
$\bmu_1$, and the time indices corresponds to $\tau_k=\frac{k}{8}$.
It is interesting to observe the smooth shift of the energy from one
``channel'' to the other one over the geodesic path while the peak shifts from one frequency to another.

\begin{figure}[htb]\begin{center}
\hspace*{-.5cm}
\includegraphics[totalheight=5.5cm]{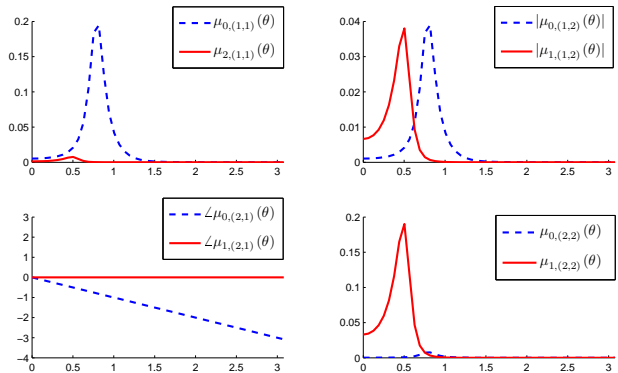}
\caption{Subplots (1,1), (1,2) and (2,2) show $\bmu_i(1,1), |\bmu_i(1,2)|$ (same as $|\bmu_i(2,1)|$) and $ \bmu_i(2,2)$. Subplot (2,1) shows $\angle(\bmu_i(2,1))$ for $i\in\left\{ 0,1\right\}$ in blue and red, respectively.}\label{fig:Mui}\end{center}
\end{figure}

\begin{figure}[htb]\begin{center}
\hspace*{-1cm}
\includegraphics[totalheight=5.5cm]{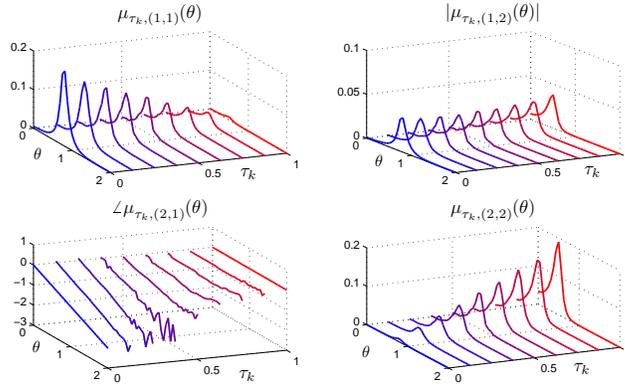}
\caption{The interpolated results $\bmu_{\tau_k}$ for $k=0,\ldots,8$ computed from \eqref{eq:interp} with $\bmu_0$ and $\bmu_1$ as the two boundary points: subplots (1,1), (1,2) and (2,2) show $\bmu_{\tau_k}(1,1), |\bmu_{\tau_k}(1,2)|$ (same as $|\bmu_{\tau_k}(2,1)|$) and $ \bmu_{\tau_k}(2,2)$, subplot (2,1) shows $\angle(\bmu_{\tau_k}(2,1))$.}\label{fig:3D}\end{center}
\end{figure}

\section{Conclusions}

This paper considers the optimal mass transportation problem of
matrix-valued densities. This is motivated by the need for a suitable topology for the spectral
analysis of multivariable time-series. It is well known that the OMT
between scalar densities induces a Riemannian metric
\cite{Benamou2000computational,Jordan1998variational} (see also
\cite{Tannenbaum2010signals} a systems viewpoint and connections to
image analysis and metrics on power spectra).
Our interest has been in extending such a Riemannian structure to
matrix-valued densities. Thus, we formulate a ``non-commutative'' version
of the Monge-Kantorovich transportation problem which can be cast as
a convex-optimization problem. Interestingly, in contrast to the
scalar case, the optimal transport plan is no longer supported on a set of
measure zero. Versions of non-commutative Monge-Kantorovich transportation has been
studied in the context of free-probability \cite{Biane2001free}. The relation of
that to our formulation is still unclear. Finally, we note that if the
matrix-valued distributions commute, then it is easy to check that
our set-up reduces to that of a number of scalar problems, which is also the case in \cite{Biane2001free}.

\section{Appendix: proof of Proposition \ref{thm:M}}

We need to prove that if $\bm(x_1, y_1)\neq 0$ and $\bm(x_2, y_2)\neq 0$, then $x_2>x_1$, $y_1>y_2$ implies
\begin{align}\label{eq:area}
(x_2-x_1)(y_1-y_2)\leq 4\lambda.
\end{align}
Without loss of generality, let
\begin{eqnarray}\label{eq:optimalm}
\bm(x_i, y_j)&=& m_{ij}\cdot A_{ij}\otimes B_{ij}
\end{eqnarray}
with $A_{ij}, B_{ij}\geq 0$, $\trace(A_{ij})=\trace(B_{ij})=1$ and $i,j \in \{1, 2\}$.
Note that $m_{12}$ and $m_{21}$ could be zero if $\bm$ does not have support on the particular point.
We assume that the condition in the proposition fails and
\begin{equation}\label{eq:assumption}
(x_2-x_1)(y_1-y_2)> 4\lambda,
\end{equation}
then we show that
by rearranging mass the cost can be reduced.

We first consider the situation when $m_{22}\geq m_{11}$. By rearranging the value of $\bm$ at the four points $(x_i, y_j)$ with $i, j\in \{1, 2\}$, we construct a new transportation plan $\tilde \bm$ at these four locations as follows
\begin{subequations}\label{eq:tildem}
\begin{eqnarray}
\tilde\bm(x_1,y_1)&=& 0\\
\tilde\bm(x_1,y_2)&=& (m_{11}+m_{12}) \cdot\tilde A_{12}\otimes \tilde B_{12}\\
\tilde\bm(x_2,y_1)&=& (m_{11}+m_{21})\cdot \tilde A_{21}\otimes \tilde B_{21}\\
\tilde \bm(x_2, y_2)&=&  (m_{22}-m_{11}) \cdot A_{22}\otimes B_{22}
\end{eqnarray}
\end{subequations}
where
\begin{eqnarray*}
\tilde A_{12}&=& \frac{m_{11}A_{11}+m_{12}A_{12}}{m_{11}+m_{12}}, \tilde B_{12}= \frac{m_{11}B_{22}+m_{12}B_{12}}{m_{11}+m_{12}}\\
\tilde A_{21}&=& \frac{m_{11}A_{22}+m_{21}A_{21}}{m_{11}+m_{21}}, \tilde B_{21}= \frac{m_{11}B_{11}+m_{21}B_{21}}{m_{11}+m_{21}}.
\end{eqnarray*}
This new transportation plan $\tilde \bm$ has the same marginals as $\bm$ at $x_1, x_2$ and $y_1, y_2$.
The original cost incurred by $\bm$ at these four locations is
\begin{equation}\label{eq:costm}
\sum_{i=1}^2\sum_{j=1}^2 m_{ij}\left( (x_i-y_j)^2+\lambda \|A_{ij}-B_{ij} \|_{\rm F}^2\right)
\end{equation}
while the cost incurred by $\tilde \bm$ is
\begin{align}\label{eq:costtildem}
&(m_{11}+m_{12})\left((x_1-y_2)^2+\lambda \|\tilde A_{12}-\tilde B_{12}\|_{\rm F}^2\right)\nonumber\\
+&(m_{11}+m_{21})\left((x_2-y_1)^2+\lambda \|\tilde A_{21}-\tilde B_{21} \|_{\rm F}^2 \right) \nonumber\\
+&(m_{22}-m_{11})\left((x_2-y_2)^2+\lambda \|A_{22}-B_{22} \|_{\rm F}^2    \right) .
\end{align}
After simplification, to show that \eqref{eq:costm} is larger than \eqref{eq:costtildem}, it suffices to show that
\begin{align}\label{eq:termA}
&~ 2m_{11} (x_2-x_1)(y_1-y_2)
\end{align}
is larger than
\begin{subequations}\label{eq:sub}
\begin{align}
&\lambda m_{11}\left(\sum_{i=1}^2\sum_{j\neq i} \|\tilde A_{ij}-\tilde B_{ij} \|_{\rm F}^2-\sum_{i=1}^2\|A_{ii}-B_{ii} \|_{\rm F}^2  \right)  \label{eq:subA}\\
&\hspace*{-.4cm}+\lambda m_{12}\left( \|\tilde A_{12}-\tilde B_{12} \|_{\rm F}^2-\|A_{12}-B_{12} \|_{\rm F}^2  \right) \label{eq:subB}\\
&\hspace*{-.4cm}+\lambda m_{21}\left( \|\tilde A_{21}-\tilde B_{21} \|_{\rm F}^2-\|A_{21}-B_{21} \|_{\rm F}^2  \right) \label{eq:subC}.
\end{align}
\end{subequations}
From the assumption in \eqref{eq:assumption}, the value of
$\eqref{eq:termA}> 20\lambda m_{11}$.
We derive upper bounds for each term in \eqref{eq:sub}. First,
\begin{eqnarray*}
\eqref{eq:subA}\leq \lambda m_{11} \left( \| \tilde A_{12}-\tilde B_{12}\|_{\rm F}^2+\| \tilde A_{21}-\tilde B_{21}\|_{\rm F}^2\right)\leq 4 \lambda m_{11}
\end{eqnarray*}
where the last inequality follows from the fact that for $A, B\geq0$ and $\trace(A)=\trace(B)=1$,
\[
\|A-B \|_{\rm F}^2=\trace(A^2-2AB+B^2)\leq \trace(A^2+B^2)\leq 2.
\]
For an upper bound of \eqref{eq:subB},
\begin{align*}
&\|\tilde A_{12}-\tilde B_{12} \|_{\rm F}^2-\|A_{12}-B_{12} \|_{\rm F}^2\\
=& \trace\left((\tilde A_{12}-\tilde B_{12}+A_{12}-B_{12})(\tilde A_{12}-\tilde B_{12}-A_{12}+B_{12})  \right)\\
=&\frac{m_{11}}{m_{11}+m_{12}}\left(\|A_{11}-B_{22}\|_{\rm F}^2-\|A_{12}-B_{12}\|_{\rm F}^2-\frac{m_{12}}{m_{11}+m_{12}}\|A_{11}-B_{22}-A_{12}+B_{12}\|_{\rm F}^2  \right)\\
\leq& \frac{m_{11}}{m_{11}+m_{12}}\|A_{11}-B_{22}\|_{\rm F}^2\\
\leq& 2\frac{m_{11}}{m_{11}+m_{12}}
\end{align*}
where the second equality follows from the definition of $\tilde A_{12}$ and $\tilde B_{12}$ while the last inequality is obtained by bounding the terms in the trace.
Thus
\[
\eqref{eq:subB}\leq 2\lambda m_{12}\frac{m_{11}}{m_{11}+m_{12}}\leq 2 \lambda m_{11}.
\]
In a similar manner, $\eqref{eq:subC}\leq 2 \lambda m_{11}$. Therefore,
\[
\eqref{eq:sub}\leq 8\lambda m_{11}< \eqref{eq:termA}
\]
which implies that the cost incurred by $\tilde \bm$ is smaller than the cost incurred by $\bm$.

For the case where $m_{11}> m_{22}$, we can prove the claim by constructing a new transportation plan $\hat\bm$ with values
\begin{eqnarray*}
\hat\bm(x_1,y_1)&=& (m_{11}-m_{22})\cdot A_{11}\otimes B_{11}\\
\hat\bm(x_1,y_2)&=& (m_{12}+m_{22}) \cdot\hat A_{12}\otimes \hat B_{12}\\
\hat \bm(x_2,y_1)&=& (m_{21}+m_{22})\cdot \hat A_{21}\otimes \hat B_{21}\\
\hat \bm(x_2, y_2)&=&  0
\end{eqnarray*}
with
\begin{eqnarray*}
\hat A_{12}&=& \frac{m_{12}A_{12}+m_{22}A_{11}}{m_{12}+m_{22}}, \hat B_{12}= \frac{m_{12}B_{12}+m_{22}B_{22}}{m_{12}+m_{22}}\\
\hat A_{21}&=& \frac{m_{21}A_{21}+m_{22}A_{22}}{m_{21}+m_{22}}, \hat B_{21}= \frac{m_{21}B_{21}+m_{22}B_{11}}{m_{21}+m_{22}}.
\end{eqnarray*}
The rest of the proof is carried out in a similar manner.

\section*{Acknowledgments}

This work was supported in part by grants from NSF, NIH, AFOSR, ONR,
and MDA.

\bibliographystyle{IEEEtran}
\bibliography{IEEEabrv,MatrixOMT}
\end{document}